\documentclass[aps,pra,groupedaddress]{revtex4}
\usepackage[english]{babel}
\usepackage{epsfig}
\usepackage{tikz}
\usepackage{hyperref}
\usepackage{amssymb, amsmath, amsthm}

\newtheorem{teo}{Theorem}[section]

\newtheorem{lema}[teo]{Lemma}
\newtheorem{prop}[teo]{Proposition}
\newtheorem{defi}[teo]{Definition}

\newcommand{\R}{\mathbb R}

\newcommand{\norm}[1]{\left\Vert#1\right\Vert}
\newcommand{\abs}[1]{\left\vert#1\right\vert}
\newcommand{\set}[1]{\left\{#1\right\}}

\newcommand{\eps}{\varepsilon}

\newcommand{\trace}[1]{\mbox{tr}\left( #1 \right)}

\newcommand{\ket}[1]{\vert#1\rangle}
\newcommand{\bra}[1]{\langle#1\vert}

\newcommand{\refe}[1]{(\ref{#1})}

\newcommand{\interior}[1]{\mbox{int}\left({#1}\right)}

\newcommand{\UNICAMP}{Departamento de Matem\'atica Aplicada, Universidade Estadual de Campinas, Campinas, SP 13083-859,
Brazil}

\begin{document}

\title{Global convergence of diluted iterations in maximum-likelihood quantum tomography }

\author{D. S. Gon\c{c}alves}
\email[]{douglas.goncalves@irisa.fr}
\affiliation{\UNICAMP}
\affiliation{IRISA, University of Rennes 1, Rennes, France}

\author{M. A. Gomes-Ruggiero}
\affiliation{\UNICAMP}

\author{C. Lavor}
\affiliation{\UNICAMP}
     
\date{\today}

\begin{abstract}
\noindent In this paper we present an inexact stepsize selection for the Diluted $R \rho R$ algorithm \cite{rehacek2007}, used to obtain the maximum likelihood estimate to the density matrix in quantum state tomography. We give a new interpretation for the diluted $R \rho R$ iterations that allows us to prove the global convergence under weaker assumptions. Thus, we propose a new algorithm which is globally convergent and suitable for practical implementation.  \\
\\
\noindent PACS number(s): 03.65.Wj
\end{abstract}
\maketitle
\section{Introduction}
In quantum state tomography, the aim is to find an estimate for the density matrix associated to the ensemble of identically prepared quantum states, based on measurement results \cite{hradil1997, parisqse2004, dgoncalves2013}. This is an important procedure in quantum information and computation, for example, to verify the fidelity of the prepared state \cite{ncbook, clavor2011} or in quantum process tomography \cite{maciel2012}. 

Besides the experimental design to get a tomographically complete set of measurements, post processing routines are required to recover information from the measurement results. Some approaches are based on direct inversion of the data while others rest in statistical based methods. For a survey, the reader can see \cite{parisqse2004}.

Among the statistical based methods, the Maximum Likelihood estimation (ML) \cite{parisqse2004, hradil2004} has been often used by experimentalists \cite{james2001}. The Maximum Likelihood estimate for the density matrix is that one which maximizes the probability of the observed data. In \cite{hradil1997,hradil2004}, it was proposed an iterative procedure to solve the problem of the maximum likelihood estimation for density matrices. We refer to this procedure as the $R \rho R$ algorithm. The main properties of the $R \rho R$ algorithm are: keeping the positivity and unit trace of the iterates and its low computational cost, involving only matrix products at each iteration. 

Although in practice the $R \rho R$ method works in most of the cases, there is no theoretical guarantee of convergence, regardless the dataset and the initial point. In \cite{rehacek2007}, the authors presented an example where the method gets into a cycle. In the same work, they proposed some kind of relaxation of the $R \rho R$ iterations, controlling the step size at each iteration by a positive parameter $t$. They called this kind of iterations as Diluted $R \rho R$ iterations. It was proved that the Diluted $R \rho R$ method converges to the maximum likelihood solution if, at each iteration, the optimal value of the step size $t$ is chosen. 

However, to find the optimal value of $t$ means to solve another optimization problem at each iteration, which  represents an undesirable additional computational cost in practice. This issue was remarked in \cite{rehacek2007}, where it was suggested some heuristics in order to get some reasonable guess for the step size $t$ in practical implementations, but loosing the convergence warranty. 

In this work we propose a new stepsize selection procedure which is reliable and feasible in practice. We give a new interpretation to the Diluted $R \rho R$ iteration, where the search direction is a combination of two ascent directions controlled by the step size $t$. This allow us to apply an inexact line search to determine the step length. Instead of the optimal value of $t$, at each iteration, it is enough to find a value which ensures a sufficient improvement in the likelihood function in order to prove the global convergence. We propose an algorithm, using an Armijo-like condition \cite{bertsekas1999,nocedal1999} and a backtracking procedure, and prove that it is globally convergent and also computationally practicable. 

This paper is organized as follows. Section \ref{theoryrrhor} reviews the theory of the $R \rho R$ algorithms for quantum tomography. The concepts of nonlinear optimization used to prove the convergence of the Diluted $R \rho R$ are presented in Section \ref{globaltheory}.  Section \ref{main} presents the proof of global convergence of the Diluted $R \rho R$ algorithm under line search and Armijo condition. Examples illustrating the differences and similarities of our proposal to the traditional fixed step length are presented in Section \ref{examples}. Section \ref{final} closes this work with some final considerations. 

\section{$R \rho R$ iterations for quantum tomography}\label{theoryrrhor}
In this section we address the theory and motivation behind the $R \rho R$ and the Diluted $R \rho R$ algorithms, following the references \cite{hradil1997,hradil2004,rehacek2007}. 

Here we consider measurements described by a POVM set $\set{E_i}_i$, where $E_i$ are semidefinite positive operators which sum to the identity. The relation between the density matrix $\rho$ and the probability outcomes is given by the Born's rule \cite{parisqse2004}:
$$
p_i(\rho) = \trace{E_i \rho}.
$$ 
Linear inversion methods equate the predicted probabilities $\set{p_i(\rho)}_i$ with the experimental data $\set{f_i}_i$:
$$
f_i = \trace{E_i \rho}, \forall i
$$ 
and the inversion of these linear equations gives an estimate of $\rho$. The main problem with this approach is that, in general, the frequencies are noisy and this fact can leads to a matrix $\rho$ outside the density matrix space (the Hermitian semidefinite positive trace one matrices). 

Among the statistical based methods, the Maximum Likelihood estimation \cite{parisqse2004, hradil2004} has been often used by experimentalists \cite{james2001}. Let us denote $\rho^{\dagger}$ the conjugate transpose of $\rho$ and $\rho \succeq 0$ to say that the Hermitian matrix $\rho$ is semidefinite positive (or $\rho \succ 0$ for a strictly positive matrix). The ML estimation searches within the density matrix space:
$$
{\cal S} = \set{\rho\ |\ \rho = \rho^{\dagger}, \ \rho \succeq 0, \trace{\rho}=1},
$$
that one which maximizes the likelihood function. The likelihood function is the probability of getting the observed data given the density matrix $\rho$. A common used likelihood \cite{hradil1997,hradil2004}, for a given data set $\set{f_i}$, is
$$
{\cal L}(\rho) \propto \prod_i p_i(\rho)^{N f_i},
$$
and since the log-likelihood is more tractable, our goal is to find $\rho$ that solves the problem
\begin{equation}
\label{mlprob}
\begin{aligned}
\max_{\rho} & \ \ \sum_i f_i \log p_i(\rho) & \equiv \ F(\rho) \\
\mbox{s.t} & \ \ \trace{\rho} = 1 \\
\ & \ \ \rho \succeq 0.
\end{aligned}
\end{equation}
\noindent The maximization of the objective function $F(\rho)$ in \refe{mlprob} is constrained to the density matrix space ${\cal S}$ which is the intersection of the semidefinite positive cone $\rho \succeq 0$ with the affine subspace $\trace{\rho}=1$. The constraints may motivate one to try semidefinite programming (SDP) methods \cite{klerkbook2002} for solving \refe{mlprob}, but efficient solvers \cite{sdpt3,sedumi} are available only for linear and quadratic objective functions. 

Other methods are based on the reparameterization \cite{james2001,goncalves2012} of the matrix variable $\rho = \rho(\theta)$ in order to automatically fulfill the constraints and then to solve an unconstrained maximization problem in the new variable $\theta$. However, generic numerical optimization methods are often slow when the number of parameters $d^2$ ($d$ is the dimension of the Hilbert space) is large. 

Here, we study an alternative algorithm, proposed in \cite{hradil2004}, 
which takes advantage of the structure of the problem \refe{mlprob} and has good convergence properties. 

Consider the gradient of the objective function $F(\rho)$, given by
\begin{equation}
\label{gradr}
\nabla F(\rho) = \sum_i \frac{f_i}{\trace{E_i \rho}} E_i \equiv R(\rho),
\end{equation}
and let $\interior{\cal S}$ be the interior of ${\cal S}$, that is
$$
\interior{\cal S} = \set{\rho \in {\cal S}\ |\ \rho \succ 0}.
$$
As it was shown in \cite{hradil2004}, a matrix $\rho \in \interior{\cal S}$ solves \refe{mlprob} if it satisfies the extremal equation
\begin{equation}
\label{extremal}
R(\rho) \rho = \rho,
\end{equation}
or equivalently
\begin{equation}
\label{extremal2}
R(\rho) \rho R(\rho) = \rho.
\end{equation}
If the density matrix $\rho$ is restricted  to diagonal matrices, the equation \refe{extremal} can be solved by the expectation-maximization (EM) algorithm \cite{vardi1993}.  The EM algorithm is guaranteed  to increase the likelihood at each step and converges to a fixed point of \refe{extremal}. However, the EM algorithm cannot be applied to the quantum problem, because without the diagonal constraint it does not preserve the positivity of the density matrix. In \cite{hradil2004}, it was proposed an iterative procedure based on the equation \refe{extremal2} instead. Let $k$ be the iteration index, and so, $\rho^k$ the current approximation to the solution. An iteration of the $R \rho R$ algorithm is given by:
$$
\rho^{k+1} = {\cal N}\,R(\rho^k) \rho^k R(\rho^k),
$$
\noindent where ${\cal N}$ is the normalization constant which ensures unit trace.  

Notice that the positivity is explicitly preserved at each step. Another remarkable property of the $R \rho R$ algorithm is its computational cost: at each iteration, it is just required to compute a matrix-matrix product. This is a quite cheap iteration in contrast with the iteration of an semidefinite programming method. 

Although the $R \rho R$ algorithm is a generalization of the EM algorithm, its convergence is not guaranteed in general. In \cite{rehacek2007}, it was presented a counterexample where the method produces a cycle. For this reason, in that work was proposed the diluted iteration of the $R \rho R$ algorithm, or simply ``Diluted $R \rho R$''. 

The idea is to control each iteration step by mixing the operator $R(\rho)$ with the identity operator:
\begin{equation}
\label{dilutedit}
\rho^{k+1} = {\cal N} \left[ \frac{I + t R(\rho^k)}{1 + t} \right] \rho^k \left[ \frac{I + t R(\rho^k)}{1 + t} \right],
\end{equation}
where $t>0$ and ${\cal N}$ is the normalization constant. It is important to observe that as $t \rightarrow \infty$, the iteration tends to the original $R \rho R$ iteration. Moreover, when $t>0$ is sufficient small, it was proved that the likelihood function is strictly increased, whenever $R(\rho)\rho \ne \rho$. It was also shown that the ``Diluted $R \rho R$'' is convergent to the ML density matrix, if the initial approximation is the maximally mixed state $\rho^0=(1/d)I$  and the optimal value of $t$:
\begin{equation}
\label{els}
t = \mbox{arg}\max_{t>0} F(\rho^{k+1}(t)),
\end{equation}
\noindent is used at each iteration. 

In nonlinear optimization \cite{bertsekas1999,nocedal1999}, this is called exact line search. Though the convergence can be achieved using this procedure, in general, solving \refe{els} may be computationally  demanding. Albeit in \cite{rehacek2007} the authors proved the convergence with the exact line search, they suggest that, in practice, one could use an ad hoc scheme to determine the ``best'' value of the steplength $t$ to be used through all iterations. 

Here, instead of \refe{els}, we propose an inexact line search to determine the steplength $t$ in each iteration \refe{dilutedit}. We do not search the best possible $t>0$, but one that ensures a sufficient improvement in the log-likelihood. We prove that this procedure is well-defined and that the iterations \refe{dilutedit} converge to a solution of \refe{mlprob}, from any positive initial matrix  $\rho^0$. The implementation of the inexact line search is straightforward and we also present some examples showing the improvements, against an ad hoc fixed $t$ strategy.

\section{Global convergence theory for ascent direction methods}\label{globaltheory} 
The purpose of this section is to expose some basic concepts of nonlinear optimization which are necessary to prove the global convergence of the Diluted $R \rho R$ algorithm under an inexact line search scheme. These concepts are classical for the optimization community and are detailed in \cite{nocedal1999,bertsekas1999}. To make it easier, we have adapted these concepts using the quantum tomography notation. 

Consider the following maximization problem over the set of Hermitian matrices ${\cal H}$:
\begin{equation}
\label{maxprob}
\begin{aligned}
\max_{\rho} & \ \ F(\rho) \\
\mbox{s.t} & \ \ \rho \in \Omega,
\end{aligned}
\end{equation}
where $f: {\cal H} \rightarrow \R$ is a continuously differentiable function and $\Omega \subset {\cal H}$ is a convex set. 

Given an approximation $\rho^k$ for the solution of problem \refe{maxprob}, ascent direction methods try to improve the current objective function value generating an ascent direction $D^k$ and updating the iterate
\begin{equation}
\label{iteration}
\rho^{k+1} = \rho^k + t_k D^k,
\end{equation}
where $t_k$ is called stepsize or steplength. 
\begin{defi}
A direction $D^k$ is an {\it ascent direction} at the iterate $\rho^k$ if 
$$
\trace{\nabla F(\rho^k) D^k} > 0,
$$
and this ensures that, for a sufficient small $t_k>0$, the function value is increased. An ascent direction $D^k$ is {\it feasible} if $\rho^{k+1}$, belongs to $\Omega$ for $t_k \in (0,\eps)$, for some $\eps>0$. 
\end{defi}
One of the insights of this work is that the diluted $R \rho R$ iteration \refe{dilutedit} can be written as an  ascent direction iteration \refe{iteration} and so, using the theory of this section, we can prove the global convergence. 

The iteration \refe{iteration} can be repeated while there exists a feasible ascent direction. If at some point $\rho^*$ there is no feasible ascent direction, then $\rho^*$ is a {\it stationary point}. It is well known that every local maximizer is a stationary point, but the converse is not true in general. If the function $f$ is concave on the convex set $\Omega$, then a stationary point is also a maximizer. 

A maximization algorithm for the problem \refe{maxprob} is called {\it globally convergent} \cite{bertsekas1999, nocedal1999} if every limit point of the sequence generated by the algorithm is a stationary point, regardless the initial approximation $\rho^0$. Although feasible ascent directions ensure that, for a sufficient small $t_k > 0$, we can increase the function value, this is not enough to ensure the global convergence. The reason is that a simple increase in the objective function, $F(\rho^{k+1})>F(\rho^k)$, along an ascent direction is a too modest objective. In order to achieve local maximizers, or at least stationary points, a {\it sufficient} increase at each iteration is required. 

Of course that a natural choice for the steplength $t_k$, along the direction $D^k$, is the solution of the problem:
\begin{equation}
\label{exls}
t_k = \mbox{argmax}_{t} \ F(\rho^k + t D^k),
\end{equation}
that is called {\it exact line search}. However, finding the global maximizer of $f$ along the direction $D^k$ is itself a hard problem, and unless the function $f$ has a special structure such as a quadratic function, for instance, the computational effort is considerable. 

To avoid the considerable computational effort in the exact line search \refe{exls}, an inexact line search can be performed. A natural scheme is to consider successive stepsize reductions. Since the search is on a ascent direction, eventually for a small $t_k$, we can obtain $F(\rho^k + t_k D^k) > F(\rho^k)$. But, this simple increase can not eliminate some convergence difficulties. One possible strategy is the use of the {\it Armijo rule}, which asks for a steplength $t$ such that a sufficient improvement in the objective function is obtained:
\begin{equation}
\label{armijo}
F(\rho^k + t D^k) > F(\rho^k) + \gamma t \ \trace{\nabla F(\rho^k) D^k},
\end{equation}
where $\gamma \in (0,1)$. We can decrease the steplength $t$ until the condition \refe{armijo} is verified. There are other alternatives to the successive stepsize reduction, for instance, strategies based on quadratic or cubic interpolation \cite{nocedal1999}. 

Besides the steplength selection, requirements on the ascent directions $D^k$ are also necessary to avoid certain  problems. For example, it is not desirable to have directions $D^k$ with small norm when we are far from the solution. It is also necessary to avoid that the sequence of directions $\set{D^k}$ become orthogonal to the gradient of $f$, because, in this case, we are in directions of almost zero variation where too small or none improvement on the objective function can be reached. A general condition that avoid such problems is called {\it gradient related} condition \cite{bertsekas1999}.

\begin{defi}[Gradient related]
A sequence of directions $\set{D^k}$ is gradient related if for any subsequence $\set{\rho^k}_{k \in {\cal K}}$ that converges to a nonstationary point, the corresponding subsequence $\set{D^k}_{k \in {\cal K}}$ is bounded and satisfies
$$
\lim_{k \rightarrow \infty} \inf_{k \in {\cal K}} \trace{ \nabla F(\rho^k) D^k} > 0.
$$
\end{defi}
\noindent This condition means that $\norm{D^k}$ does not become 'too small' or 'too large'  relative to $\norm{\nabla F(\rho^k)}$, and that the $D^k$ and $\nabla F(\rho^k)$ do not become orthogonal. 

If an algorithm generates ascent directions satisfying the gradient related condition and the stepsizes are selected according to the Armijo rule, then it is possible to prove the global convergence \cite{bertsekas1999}.
\begin{teo}[Global convergence]
\label{teogc}
Let $\set{\rho^k}$ be a sequence generated by a feasible ascent directions method $\rho^{k+1} = \rho^k + t_k D^k$, and assume that $\set{D^k}$ is gradient related and $t_k$ is chosen by the Armijo rule. Then, every limit point of $\set{D^k}$ is a stationary point.
\end{teo}
\begin{proof}
See \cite[Proposition 2.2.1]{bertsekas1999}.
\end{proof}

\section{Convergence of the Diluted $R \rho R$}\label{main}
Sections \ref{theoryrrhor} and \ref{globaltheory} gave us the necessary background to show the convergence of the diluted $R \rho R$ iterations using an inexact line search to determine the stepsize $t$. In this section, firstly we show that the diluted iteration \refe{dilutedit} can be written as an ascent direction iteration \refe{iteration}. So, we give a geometrical interpretation and prove that the corresponding sequence of directions $\set{D^k}$ is gradient related. Finally, using the Armijo condition and a backtracking procedure, we present an algorithm which is globally convergent following the Theorem \ref{teogc}. 

From now on, we will use the notation $\nabla F(\rho)$ instead of $R(\rho)$. So, the equation \refe{dilutedit} becomes:
\begin{equation}
\rho^{k+1}  = {\cal N} \left[ \frac{I + t \nabla F(\rho^k)}{1 + t} \right] \rho^k \left[ \frac{I + t \nabla F(\rho^k)}{1 + t} \right]
\end{equation}
or 
\begin{equation}
\rho^{k+1} = \frac{(I + t \nabla F(\rho^k) )\, \rho^k \,(I + t \nabla F(\rho^k) )}{\trace{(I + t \nabla F(\rho^k) )\, \rho^k \,(I + t \nabla F(\rho^k) )}} \equiv G(\rho^k).
\end{equation} 

The expression above can be seen as a fixed point iteration. Expanding that expression, we obtain
\begin{equation}
\label{gdef}
G(\rho) = \frac{\rho + t\,(\nabla F(\rho) \rho + \rho \nabla F(\rho) ) + t^2\, \nabla F(\rho)\rho \nabla F(\rho)}{1 + 2t + t^2 \trace{\nabla F(\rho) \rho \nabla F(\rho)}}.
\end{equation}

Notice that $\rho^*$ is a fixed point of $G(\rho)$, $G(\rho^*)=\rho^*$, for $t>0$, if the following conditions are satisfied
\begin{eqnarray}
\rho^* & = & \nabla F(\rho^*) \rho^* \nabla F(\rho^*) = \nabla F(\rho^*) \rho^*.
\end{eqnarray}
If the above conditions are verified at a positive definite trace one matrix $\bar{\rho}$, then the {\it optimality conditions} \cite{bertsekas1999,nocedal1999} for the problem \refe{mlprob} are satisfied and $\rho^*$ is the maximum likelihood estimate. 

The following two lemmas are useful when concerning the $R \rho R$ iterations.

\begin{lema}\label{lemtrgradrho}
For all $\rho \in \interior{{\cal S}}$, we have
$$
\trace{\nabla F(\rho)\rho} = 1.
$$
\end{lema}
\begin{proof}
Directly from \refe{gradr}.
\end{proof}
\begin{lema}\label{lemtrgradrhograd}
If $\rho \in \interior{{\cal S}}$, then
$$
\trace{\nabla F(\rho)\rho \nabla F(\rho)} \ge 1,
$$
with equality if and only if $\rho = \nabla F(\rho) \rho$.
\end{lema}
\begin{proof} From Lemma \ref{lemtrgradrho},
$$
1 = \trace{\nabla F(\rho) \rho} = \trace{\nabla F(\rho) \rho^{1/2} \rho^{1/2}},
$$
and from the Cauchy-Schwarz inequality, we obtain
$$
1 = \abs{\trace{\nabla F(\rho) \rho^{1/2} \rho^{1/2}}}^2 \le \trace{\nabla F(\rho)\rho \nabla F(\rho)}\trace{\rho} = \trace{\nabla F(\rho)\rho \nabla F(\rho)}.
$$
\noindent The equality in Cauchy-Schwarz occurs when $\nabla F(\rho) \rho^{1/2} = \alpha \rho^{1/2}$, or  equivalently, when $\nabla(\rho) \rho = \rho$.
\end{proof}
\noindent Let us simplify the expression \refe{gdef}, defining for some $\rho$,
$$
q(t) = 1 + 2t + t^2 \trace{\nabla F(\rho) \rho \nabla F(\rho)}.
$$
Since $\trace{\nabla F(\rho) \rho \nabla F(\rho)} \ge 1$, for any density matrix $\rho$, we have that $q(t) \ge 1$ for all $t \ge 0$. Furthermore, if $\rho \in {\cal S}$, the set of density matrices, $G(\rho) \in {\cal S}$ as well, for any $t \ge 0$. Thus, $G(\rho)$ defines a path on the density matrices space ${\cal S}$, parameterized by $t$ such that, when $t \rightarrow 0$, $G(\rho) \rightarrow \rho$, and when $t \rightarrow \infty$, $G(\rho) \rightarrow \tilde{\rho}$, where
$$
\tilde{\rho} = \frac{\nabla F(\rho) \rho \nabla F(\rho)}{\trace{\nabla F(\rho) \rho \nabla F(\rho)}},
$$
as in the original $R \rho R$ algorithm \cite{hradil1997}. 

Let us also define the point
\begin{equation}
\label{rhobar}
\bar{\rho} = \frac{\nabla F(\rho) \rho + \rho \nabla F(\rho)}{2}.
\end{equation}
Unlike the point $\tilde{\rho}$, the point $\bar{\rho}$, in general, is not in the set ${\cal S}$. 

Now, rewriting the expression \refe{gdef}, we obtain
\begin{equation}
\label{gexp}
G(\rho) = \frac{1}{q(t)} \rho + \frac{2t}{q(t)} \left( \frac{\nabla F(\rho) \rho + \rho \nabla F(\rho) }{2} \right) + \frac{t^2 \trace{\nabla F(\rho) \rho \nabla F(\rho)} }{q(t)} \frac{\nabla F(\rho) \rho \nabla F(\rho)}{\trace{\nabla F(\rho) \rho \nabla F(\rho)}},
\end{equation}
that is
$$
G(\rho) = \frac{1}{q(t)} \rho + \frac{2t}{q(t)} \bar{\rho} + \frac{t^2 \trace{\nabla F(\rho) \rho \nabla F(\rho)} }{q(t)} \tilde{\rho}.
$$
Therefore, we have a convex combination of the points $\rho,\ \bar{\rho}$, $\tilde{\rho}$, and the path defined by $t$ is in the convex set whose extreme points are $\rho,\ \bar{\rho}$ and $\tilde{\rho}$, as we can see in Figure \ref{figgeo}. \\
\begin{figure}[!h]
\centering
\begin{tikzpicture}[scale=0.7]
\filldraw[draw=black, fill=white] (0,0) circle (80pt);
\filldraw (40pt,-30pt) circle (2pt) node[below]{$\rho^k$};
\filldraw (100pt,0pt) circle (2pt) node[below]{$\bar{\rho}^k$};
\filldraw (60pt,30pt) circle (2pt) node[above]{$\tilde{\rho}^k$};
\draw[densely dotted] (40pt,-30pt) -- (100pt,0pt) -- (60pt,30pt) -- (40pt,-30pt);
\draw[->,gray,thick] (40pt,-30pt) -- (98pt,-1pt);
\draw[->,gray,thick] (40pt,-30pt) -- (59pt,27pt);
\draw (40pt,-30pt) .. controls (100pt,0pt) and (70pt,22pt) .. (60pt,30pt);
\draw (75pt,-15pt) node[below]{$\bar{D}^k$};
\draw (40pt,20pt) node[below]{$\tilde{D}^k$};
\draw (-170pt,-100pt) rectangle (170pt,100pt);
\draw (-160pt,-90pt)[right]node {$\trace{\rho}=1$};
\draw[->] (50pt,40pt) -- (30pt,50pt) node[left]{{\small $(t \rightarrow \infty)$}};
\draw[->] (30pt,-35pt) -- (20pt,-40pt) node[left]{{\small $(t = 0)$}};
\draw[->] (75pt,15pt) -- (95pt,35pt) node[right]{{\small $\rho^{k+1}(t)$}};

\draw (0pt,0pt) node {${\cal S}$};
\end{tikzpicture}
\caption{Geometrical interpretation of $G(\rho^k)$ as a curved path parametrized by $t$.}
\label{figgeo}
\end{figure}
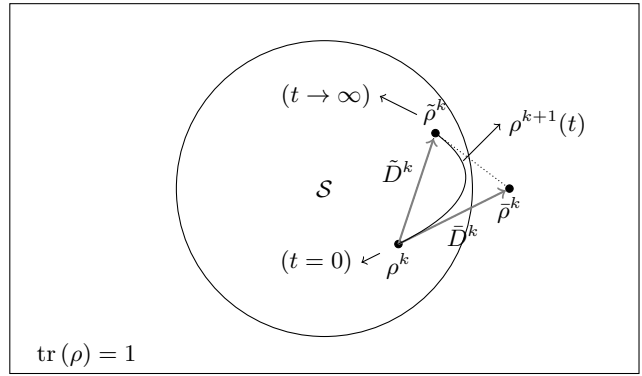

Finally, defining the directions
\begin{eqnarray}
\bar{D} & = & \bar{\rho} - \rho = \frac{\nabla f (\rho) \rho + \rho \nabla F(\rho) }{2} - \rho, \label{dbar} \\
\tilde{D} & = & \tilde{\rho} - \rho = \frac{\nabla F(\rho) \rho \nabla F(\rho)}{\trace{\nabla F(\rho) \rho \nabla F(\rho)}} - \rho, \label{dtilde}
\end{eqnarray}
and using \refe{gexp}, we obtain
$$
G(\rho) = \rho + \frac{2t}{q(t)}\bar{D} + \frac{t^2 \trace{\nabla F(\rho) \rho \nabla F(\rho)}}{q(t)} \tilde{D},
$$
which provides us an iteration like \refe{iteration}
$$
\hat{\rho} = G(\rho) = \rho + t D,
$$
where
\begin{equation}
\label{direction}
D = \frac{2}{q(t)}\bar{D} + \frac{t \, \trace{\nabla F(\rho) \rho \nabla F(\rho)}}{q(t)} \tilde{D}.
\end{equation}
The search direction $D$ is a combination of the directions $\bar{D}$ and $\tilde{D}$ controlled by the parameter $t$. From Figure \ref{figgeo}, we can see that as $t \rightarrow \infty$, $D$ goes to the direction $\tilde{D}$,  whereas $t \rightarrow 0$, $D$ becomes parallel to $\bar{D}$. It is worth to prove that these are feasible ascent directions. 
\begin{prop}
The directions $\bar{D}$ and $\tilde{D}$ are feasible ascent directions for any nonstationary point $\rho$.
\end{prop}
\begin{proof}
To prove that $\bar{D}$ is an ascent direction, we need to show that $\trace{\nabla F(\rho) \bar{D}}>0$. Using the definition of $\bar{D}$, we get
$$
\bar{D} = \bar{\rho} - \rho = \frac{\nabla F(\rho) \rho + \rho \nabla F(\rho)}{2} - \rho.
$$
For a nonstationary $\rho$ ($\nabla F(\rho) \rho \ne \rho$),
$$
\trace{\nabla F(\rho) \bar{D}} = \trace{\nabla F(\rho) \rho \nabla F(\rho) } - \trace{\nabla F(\rho) \rho} =
$$
$$
\trace{\nabla F(\rho) \rho \nabla F(\rho) } - 1 > 0,
$$
by the Cauchy-Schwarz inequality, which implies that $\bar{D}$ is an ascent direction. If $\rho \in \interior{\cal S}$, then there exists $t>0$ such that $\rho + t \bar{D} \in {\cal S}$, so the direction is feasible. 

In a similar way for $\tilde{D}$,
$$ 
\trace{\nabla F(\rho) \tilde{D}} = \frac{ \trace{\nabla F(\rho) \nabla F(\rho) \rho \nabla F(\rho) } }{\trace{\nabla F(\rho) \rho \nabla F(\rho) }} - 1 > 0.
$$
Using the fact that $\tilde{\rho} \in {\cal S}$, for $t \in (0,1]$, we get $\rho + t \tilde{D} \in {\cal S}$ as well.
\end{proof}
\noindent Since the direction $D$ is a positive combination of feasible ascent directions, it is also a feasible ascent direction. 

\noindent Now, if we can show that the sequence of directions $\set{D^k}$ is gradient related, then we can prove the global convergence under an inexact line search scheme. First, we present some technical lemmas which are useful to show the desired result. 

\begin{lema}\label{lemarhobar}
For $\rho^k \succ 0$ and $\trace{\rho^k}=1$, the matrix $\bar{\rho}^k$, defined in \refe{rhobar}, is the solution of the problem
\begin{equation}\label{problemarhobar}
\begin{aligned}
\max_{\rho} & \ \ \  \trace{\nabla F(\rho^k) (\rho - \rho^k)} - \frac{1}{2}\trace{(\rho - \rho^k)({\rho^k})^{-1}(\rho - \rho^k)} \\
\mbox{s.t} & \ \ \ \trace{\rho} = 1.
\end{aligned}
\end{equation}
\end{lema}
\begin{proof}
Consider the optimality conditions for \refe{problemarhobar}:
\begin{equation}\label{kkt1}
-\nabla F(\rho^k) + \frac{1}{2}\left[(\rho - \rho^k)(\rho^k)^{-1} + (\rho^k)^{-1}(\rho - \rho^k)\right] + \lambda_0 I  = 0,
\end{equation}
\begin{equation}\label{kkt2}
\trace{\rho} = 1.
\end{equation}
In equation \refe{kkt1}, multiplying at the right by $\rho^k$ and taking the trace, we have
$$
-\trace{ \nabla F(\rho^k) \rho^k} + \trace{\rho - \rho^k} + \lambda_0 \trace{\rho^k} = 0,
$$
which implies that $\lambda_0 = 1$. So, from
$$
-\nabla F(\rho^k) + \frac{1}{2}\left[(\rho - \rho^k)(\rho^k)^{-1} + (\rho^k)^{-1}(\rho - \rho^k)\right] + I  = 0,
$$
we obtain
$$
\rho (\rho^k)^{-1} + (\rho^k)^{-1} \rho = \nabla F(\rho^k).
$$
Using the symmetry of the solution $\rho$, the symmetry of $\rho^k$, and $\nabla F(\rho^k)$, we conclude that
$$
\rho = \nabla F(\rho^k) \rho^k = \rho^k \nabla F(\rho^k) = \frac{\nabla F(\rho^k) \rho^k + \rho^k \nabla F(\rho^k)}{2} = \bar{\rho}^k.
$$
\end{proof}

\begin{lema}\label{lemagr}
The sequence of directions $\set{\bar{D}^k}$, used to define the sequence $\set{\rho^k}$ by
$$
\rho^{k+1} = \rho^k + t\, \bar{D}^k,
$$
satisfies
$$
\lim_{k \rightarrow \infty} \inf_{k \in {\cal K}} \trace{\nabla F(\rho^k)(\bar{\rho}^k - \rho^k)} > 0,
$$
for all subsequence $\set{\rho^k}_{k \in {\cal K}}$ that converges to a non-stationary point $\rho'$.
\end{lema}
\begin{proof}
Suppose there is a subsequence $\set{\rho^k}_{k \in {\cal K}}$ that converges to a non-stationary point $\rho'$. Lemma \ref{lemarhobar} tell us that $\bar{\rho}^k$ is the solution of \refe{problemarhobar}. Thus, at  $\bar{\rho}^k$, the gradient of the objective function of \refe{problemarhobar} is orthogonal to the hyperplane $\trace{\rho}=1$, that is
$$
\trace{\left[\nabla F(\rho^k) - \frac{1}{2}\left( (\bar{\rho}^k - \rho^k)(\rho^k)^{-1} + (\rho^k)^{-1} (\bar{\rho}^k - \rho^k) \right) \right] (\rho - \bar{\rho}^k)} = 0, 
$$
\noindent $\forall \rho \mbox{ \ such that \ } \trace{\rho}=1$. Since the feasible set of \refe{problemarhobar} contains ${\cal S}$, we have
$$
\trace{\left[\nabla F(\rho^k) - \frac{1}{2}\left( (\bar{\rho}^k - \rho^k)(\rho^k)^{-1} + (\rho^k)^{-1} (\bar{\rho}^k - \rho^k) \right) \right] (\rho - \bar{\rho}^k)} = 0, \ \forall \rho \in {\cal S}.
$$
\noindent Expanding the last expression, we obtain
$$
\trace{\nabla F(\rho^k)(\rho - \bar{\rho}^k)} = - \frac{1}{2}\left[ \trace{(\rho^k - \bar{\rho}^k) (\rho^k)^{-1} (\rho - \bar{\rho}^k)} + \trace{(\rho - \bar{\rho}^k) (\rho^k)^{-1} (\rho^k - \bar{\rho}^k)} \right], 
$$
$\noindent \forall \rho \in {\cal S}$. In particular, for $\rho = \rho^k$,
\begin{equation}
\label{keyp}
\trace{\nabla F(\rho^k)(\bar{\rho}^k - \rho^k)} =  \trace{(\rho^k - \bar{\rho}^k) (\rho^k)^{-1} (\rho^k - \bar{\rho}^k)} = \norm{\rho^k - \bar{\rho}^k}_{(\rho^k)^{-1}}^2.
\end{equation}
Using the continuity of the solution given by Lemma \refe{lemarhobar}, we have
$$
\lim_{k \rightarrow \infty, \ k \in {\cal K}} \bar{\rho}^k = \bar{\rho} = \frac{\nabla F(\rho')\rho' + \rho' \nabla F(\rho')}{2}.
$$
Taking limits in \refe{keyp}, we obtain
$$
\lim_{k \rightarrow \infty} \inf_{k \in {\cal K}} \trace{\nabla F(\rho^k)(\bar{\rho}^k - \rho^k)} = \norm{\rho'  - \bar{\rho}}_{(\rho')^{-1}}^2 > 0.
$$
Since $\rho'$ is non-stationary, the right hand side of the above inequality is strictly positive and this completes the proof.
\end{proof} 

Finally, using the previous lemmas, we can prove the main assertion of this section.

\begin{prop}\label{main_prop}
The sequence of directions $\set{D^k}$ is gradient related.
\end{prop}
\begin{proof}
First, let us show that $\set{D^k}$ is bounded. In fact, $\rho^{k+1}(t_k) = \rho^k + t_k D^k = G(\rho^k)$ is in ${\cal S}$, since $\rho^k \succ 0$ and $t_k \ge 0$, by definition. In particular, for $t_k=1$, we have $\rho^{k+1}(1) = \rho^k + D^k \in {\cal S}$, and since ${\cal S}$ is bounded, then $\set{D^k}$ is also bounded. 

Now, let $\set{\rho^k}_{k \in {\cal K}}$ be a subsequence of the sequence $\set{\rho^k}$ generated by the iterations $\rho^{k+1} = \rho^k + t_k D^k$. Suppose $\set{\rho^k}_{k \in {\cal K}}$ converges to a nonstationary point $\rho'$. Using the definition of $D^k$, we obtain
$$
\trace{\nabla F(\rho^k) D^k} = \frac{2}{q(t_k)} \trace{\nabla F(\rho^k)\bar{D}^k} + \frac{t_k \trace{\nabla F(\rho^k) \rho^k \nabla F(\rho^k) }}{q(t_k)} \trace{\nabla F(\rho^k) \tilde{D}^k}.
$$
The second term in the right hand side is nonnegative, then
$$
\trace{\nabla F(\rho^k) D^k} \ge \frac{2}{q(t_k)} \trace{\nabla F(\rho^k)\bar{D}^k}.
$$
Considering $t_k \in (0,t_{max}]$, we have
$$
\trace{\nabla F(\rho^k) D^k} \ge \frac{2}{q(t_{max})} \trace{\nabla F(\rho^k)\bar{D}^k}.
$$
Taking the limit for a subsequence converging to a nonstationary point,
$$
\lim_{k \rightarrow \infty} \inf_{k \in {\cal K}} \trace{\nabla F(\rho^k) D^k} \ge \frac{2}{q(t_{max})} \lim_{k \rightarrow \infty} \inf_{k \in {\cal K}} \trace{\nabla F(\rho^k)\bar{D}^k},
$$
and since $\set{\bar{D}^k}$ is gradient related, by Lemma \ref{lemagr},
$$
\lim_{k \rightarrow \infty} \inf_{k \in {\cal K}} \trace{\nabla F(\rho^k) D^k} > 0,
$$
which implies that $\set{D^k}$ is gradient related.
\end{proof}
\noindent Thus, choosing the step size $t_k$ at each iteration, such that the Armijo condition \refe{armijo} is fulfilled, we obtain a globally convergent algorithm following the Theorem \ref{teogc}. \\
\ \\
In this way, we can define the steps of a globally convergent algorithm that uses an inexact line search as the  following: \\
\ \\
{\bf Algorithm~1} \\
\ \\
{\bf Step 0.} Given $\rho^0 \succ 0$ such that $\trace{\rho^0}=1$, $t_{max} > 0$ and $0<\alpha_0<\alpha_1<1$, set $k=0$ and $t_0 = t_{max}$. \\
\ \\
{\bf Step 1.} If some stopping criterion is verified, stop. Otherwise, compute the directions $\bar{D}^k$ and $\tilde{D}^k$, defined in \refe{dbar} and \refe{dtilde}. Set $t=\max \set{1,\ t_{k-1}}$. \\
\ \\
{\bf Step 2.} Set
$$
D = \left( \frac{2}{q(t)}\bar{D}^k + \frac{t \, \trace{\nabla F(\rho^k) \rho^k \nabla F(\rho^k)}}{q(t)} \tilde{D}^k \right).
$$
If
$$
F(\rho^k + t D) \le F(\rho^k) + \gamma\,t\, \trace{\nabla F(\rho^k) D},
$$
choose $t \in [\alpha_0\,t, \ \alpha_1\,t]$ and go to Step 2. \\
\ \\
{\bf Step 3.} Set $t_k = t$, $D^k = D$ and $\rho^{k+1} = \rho^k + t_k\, D^k$. Go to the step 1.\\
\ \\
The Theorem \ref{mainteo} states the desired result, that is, any limit point of the sequence generated by Algorithm~1 is a stationary point, regardless the initial approximation. Since the problem \refe{mlprob} is convex, then a stationary problem is also a solution.

\begin{teo}\label{mainteo}
Every limit point $\rho^*$ of a sequence $\set{\rho^k}$, generated by the Algorithm~1, is a stationary point, that is, $\nabla F(\rho^*) \rho^* = \rho^* = \rho^* \nabla F(\rho^*)$. 
\end{teo}
\begin{proof}
Using the Proposition \ref{main_prop}, we have that $\set{D^k}_k$, used in Algorithm 1, is gradient related. Since the step selection in Algorithm 1 satisfies the Armijo condition, then we can apply the Theorem \ref{teogc} to obtain the claimed result.
\end{proof}
\noindent In the step 2 of Algorithm~1, instead of successive reductions of the steplength $t$, one could use, for instance, a quadratic or cubic interpolation \cite{nocedal1999} to estimate $t$ that maximizes $F(\rho^k + t D^k)$, in order to turn the search more effective.

\section{Illustrative examples}\label{examples}
In this section we selected two illustrative examples to show that Algorithm~1 outperforms the Diluted $R \rho R$ algorithm with fixed stepsize \cite{rehacek2007}. Besides Algorithm~1 converges in problems where the original $R \rho R$ does not, it also reduces the number of iterations when compared to the fixed stepsize version of the Diluted $R \rho R$, without harming the convergence behavior in cases where the last one works. 

First, we consider the counterexample where the pure $R \rho R$ method gets into a cycle \cite{rehacek2007}. Suppose we made three measurements on a qubit with the apparatus described by $\Pi_0 = \ket{0}\bra{0}$ and $\Pi_1 = \ket{1}\bra{1}$, detecting $\ket{0}$ once and $\ket{1}$ twice. We used the completely mixed state as starting point and considered convergence when the distance between two consecutive iterates is small enough (less than $10^ {-7}$). For each $t$ fixed in the Diluted $R \rho R$, we define $t_{max}=t$ in the algorithm that uses line search. We also used $\gamma=10^{-4}$ and $\alpha_0=\alpha_1=0.5$ in the Algorithm~1.
\begin{figure}[!h]
\centering
\epsfig{file=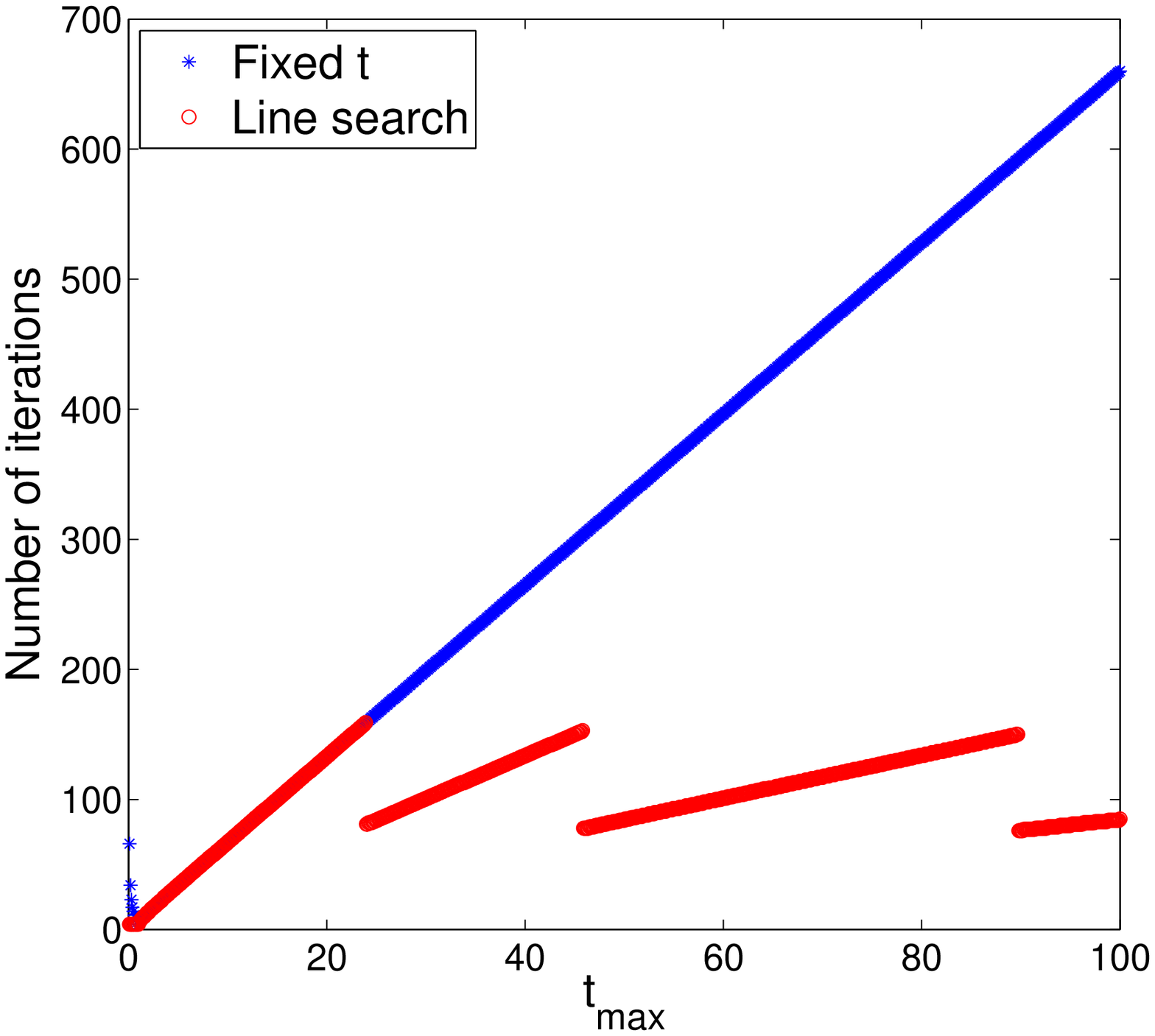,width=7cm}
\epsfig{file=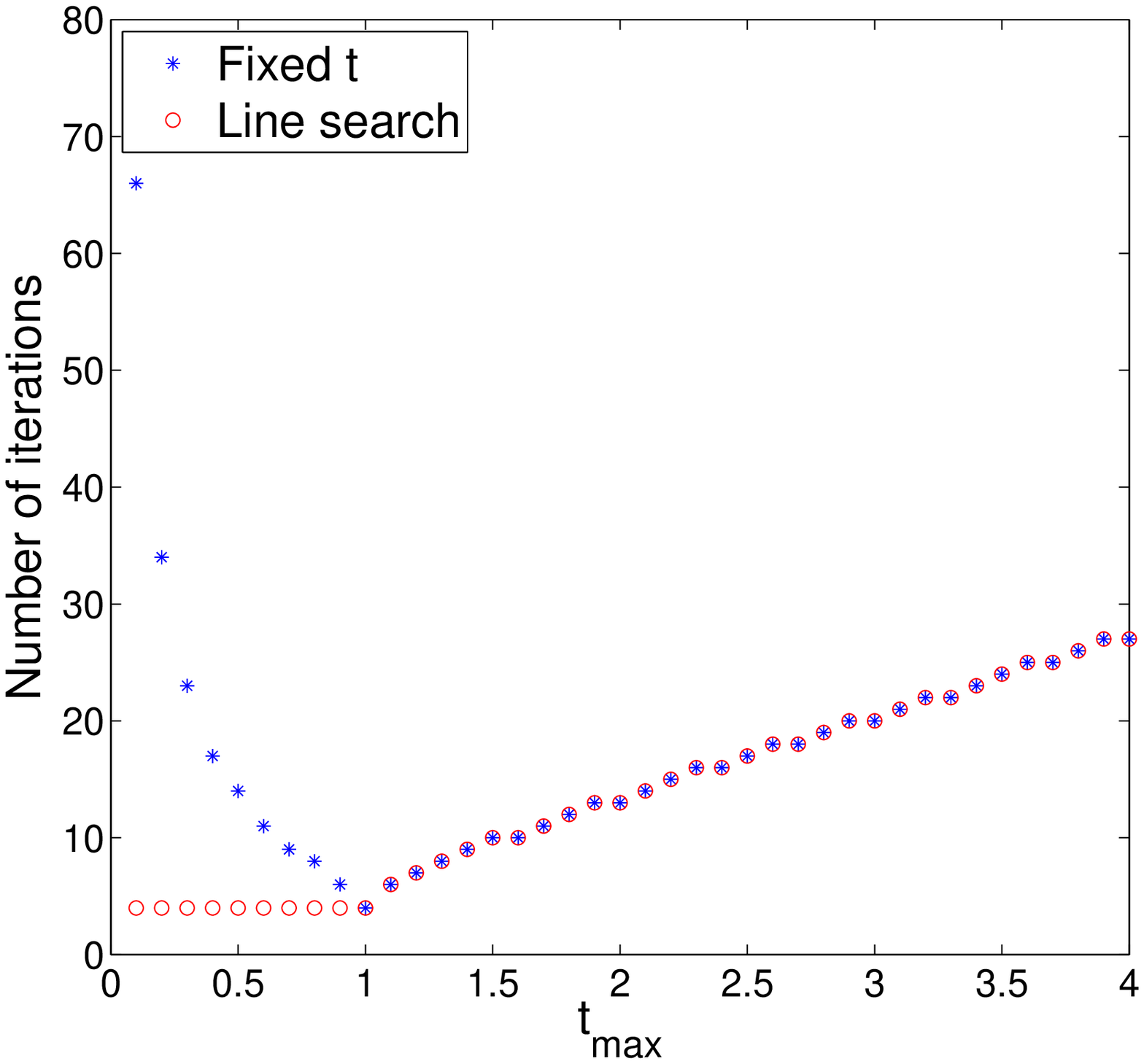,width=7cm}
\caption{Number of iterations as a function of $t$.}
\label{fig1}
\end{figure}\\
The Figure \ref{fig1} brings the comparison between the version with fixed step size (stars) against the one with line search (circles), described in the previous section. In the left panel, we can see that the number of iterations grows up as the stepsize $t$ increases, for the ``fixed  $t$'' strategy. This was expected because as $t \rightarrow \infty$, the iterations tend to be pure $R \rho R$ iterations, and in this limit case, there is no convergence. Conversely, the line search strategy keeps the number of iterations bounded, regardless the value of $t_{max}$. 

The right panel is a zoomed version of the left one, in order to show the behavior for small values of $t$. As expected, although the Diluted $R \rho R$ guarantees the monotonic increase of the likelihood for sufficient small steps, repeating too small steps leads to more iterations of the method. The Algorithm~1 ensures a substantial increase of the likelihood through the line search procedure. To avoid extremely small steps, at each iteration of the Algorithm~1, the first trial for $t_k$ is at least one.
\begin{figure}[!h]
\centering
\epsfig{file=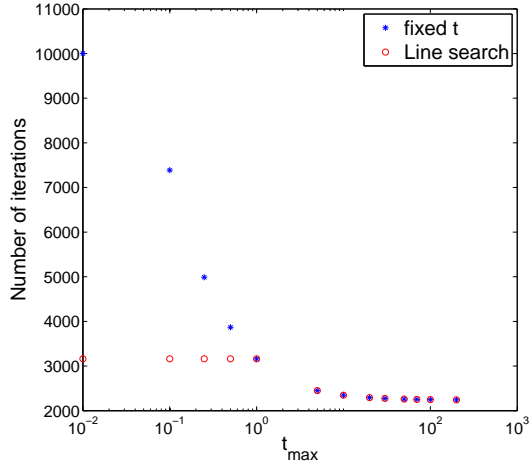,width=7.5cm}
\caption{Number of iterations as a function of $t$ (W state tomography).}
\label{fig2}
\end{figure}

Second, we consider as data the theoretical probabilities for the W state. The Figure \ref{fig2} presents the number of iterations for different values of $t$ (log-scale). Again, fixed small values of $t$ will produce a higher number of iterations. It is also important to note, in this example, that the behavior of the line search version is the same as the ``fixed $t$'' one, as the suggested step length $t_{max}$ increases. This means that in the Algorithm~1, the full step $t_k = t_{max}$ was accepted (fulfills the Armijo condition) in every iteration.  

In \cite{rehacek2007}, the authors claim that one should first try a larger value for the step size $t$ and perform Diluted $R \rho R$ iterations with the same value of $t$. If the iterations do not converge, then try a smaller value of $t$. This ad hoc procedure was motivated because the pattern of the Figure \ref{fig2} often occurs in practice, and then larger $t$ means less iterations. However, what should be a good guess for a larger value of $t$ in order to ensure few iterations? And if the convergence does not occur, how to choose a smaller value of $t$ to guarantee the convergence? These issues could result in a lot of re-runs until a good value of $t$ can be found, which can change from one dataset to another. 

These examples illustrate that the Armijo line search procedure represents an improvement on the Diluted $R \rho R$ algorithm, adjusting the step length $t$ just when necessary, and show that the convergence does not depend on a specific choice of a fixed step length or the starting point.

\section{Final remarks}\label{final}
We proved the global convergence of the Diluted $R \rho R$ algorithm under a line search procedure with Armijo condition. The inexact line search is a weaker assumption than the exact line search used in convergence proofs of a previous work \cite{rehacek2007}. Moreover, the proposed globalization by line search does not depend on 
 the guess of a fixed step length for all iterations. Instead, as usual in nonlinear optimization, the step length is adjusted just when necessary in order to ensure a sufficient improvement in the likelihood at each iteration. Thus, the Armijo line search procedure is a reliable globalization and represents a practical improvement in   the Diluted $R \rho R$ algorithm for quantum tomography.
 
 \section*{Acknowledgements}
We thanks to the Brazilian research agencies FAPESP, CNPq and INCT-IQ (National Institute for Science and Technology for Quantum Information). DG also thanks the Brittany Region (France) and INRIA for partial financial support.
\bibliographystyle{apsrev4-1}
\bibliography{steplength_diluted_rrhor}

\end{document}